\documentclass[11pt,letterpaper]{article}
\usepackage{times}
\usepackage{bbm}

\usepackage{enumitem}
\usepackage{amsmath}

\usepackage{amsthm}
\usepackage{amssymb}
\usepackage{mathtools}
\usepackage{dsfont}

\usepackage{cleveref}

\usepackage{color}

\newcommand{\ignore}[1]{}

\newtheorem{theorem}{Theorem}[section]
\newtheorem{lemma}[theorem]{Lemma}

\newtheorem{proposition}[theorem]{Proposition}

\newtheorem{corollary}[theorem]{Corollary}

\newtheorem{definition}{Definition}[section]

\usepackage{booktabs} 
\usepackage[ruled]{algorithm2e} 
\usepackage{algorithmic}

\newcommand{\budget}{budget in tokens}

\begin{document}

    \title{Competitive Equilibria with Unequal Budgets: \\
      Supporting Arbitrary Pareto Optimal Allocations\thanks{
The work of Michal Feldman and Nir Andelamn has received funding from the European Research Council (ERC) under the European Union's Horizon 2020 research and innovation program (grant agreement No. 866132), and from the Israel Science Foundation (grant number 317/17). The work of Yishay Mansour and Nir Andelamn has received funding from the European Research Council (ERC) under the European Union's Horizon 2020 research and innovation program (grant agreement No. 882396),  by the Israel Science Foundation (grant number 993/17)  and the Yandex Initiative for Machine Learning at Tel Aviv University. The work of Amos Fiat received funding from the Israel Science Foundation (grant number 1841/14) and the Blavatnik fund.}}

\author{
    Nir Andelman\thanks{Afeka; {\tt nir.andelman@gmail.com}}
    \and
    Michal Feldman\thanks{Blavatnik School of Computer Science, Tel-Aviv University; {\tt mfeldman@tauex.tau.ac.il}}
    \and
    Amos Fiat\thanks{Blavatnik School of Computer Science, Tel-Aviv University; {\tt fiat@tau.ac.il}}
    \and
    Yishay Mansour\thanks{Blavatnik School of Computer Science, Tel-Aviv University and Google Research; {\tt mansour.yishay@gmail.com}}
}

\maketitle

\begin{abstract}
We consider a market setting of agents with additive valuations over heterogeneous divisible resources. Agents are assigned a budget of tokens (possibly unequal budgets) they can use to obtain resources; leftover tokens are worthless.
We show how to support any Pareto efficient allocation in equilibrium, using anonymous resource prices and agent specific budgets.
We also give computationally efficient algorithms for those tasks. In particular, this allows us to support the Rawlsian max-min allocation.
\end{abstract}

%

\section{Introduction}

We consider a set of heterogeneous divisible resources, and a set of agents, each of which has (possibly different) valuations for the resources, and the value for a collection of [fractional] resources is additive.

Equilibrium pricing of markets appears in many guises: 

\begin{enumerate}
	\item Leon Walras studied market clearing prices (Walrasian Equilibrium), as far back as 1874 \cite{walras1874}. This is a quasi linear model where the utility of an agent is the value of the goods minus the payments used to purchase them. Conditions and algorithms to find Walrasian equilibrium prices were studied in \cite{kelso1982job,GulStacccetti} and many subsequent papers, see survey in \cite{renatosurvey}. 
	\item Irving Fisher's 1891 PhD thesis \cite{BrainardScarf} introduced price equilibrium for Fisher Markets. The utility is the sum of valuations of the goods received. Tokens are used to obtain and pay for goods but the tokens themselves have no value. Every agent $i$ has a budget $b_i$ in tokens, and valuations for the various goods, the key problem studied in the literature is to compute prices so that competitive equilibria exist. For algorithmic perspectives see \cite{Vajiranifisher,CVfisher,Cole17}. 
	\item Arrow and Debreu \cite{Arrow51,Debreu51,ArrowDebreu} consider a much more general problem and 
	defined a competitive equilibrium as an assignment of prices to the resources and an associated allocation so that the allocation to every agent maximizes the agent utility over all other allocations (where money is considered one of the resources and the model need not treat money linearly)\footnote{We remark they actually address a more general setting where agents have endowments and they can both produce and consume goods.}.
	\item To study fair allocation of limited resources, Varian \cite{Varian74} considers the problem of competitive equilibria from equal incomes (CEEI). In a competitive equilibria from equal incomes   resources are assigned prices, all agents get one token to spend on purchasing such resources, and there is an associated  allocation of resources such that: (i) every agent is maximally happy subject to the prices and her budget. For CEEI any unused budget is worthless so all the initial ``income'' is used up, and all resources are distributed.
\end{enumerate}

In the settings above, anonymous prices are set for every resource thus achieving a measure of transparency. No agent ever wants to switch her allocation with that of another (which paying the same price as the other), ergo, the allocation is envy-free. Public anonymous prices are a major component of transparency in markets. Pareto-efficient is a very desirable goal of market pricing \cite{Stiglitz81}. Social welfare maximization is an example of a Pareto-efficient allocation but not the only one. 


We consider a setting with agents with additive valuations over heterogeneous divisible resources. Agents are allocated a budget of tokens (possibly unequal budgets) they can use to obtain resources; leftover tokens are worthless.
We wish to support any Pareto efficient allocation in equilibrium, using anonymous resource prices and agent specific allocation of budgets.

Our main result is the following:

\vspace{0.1in}
\noindent {\bf Main result:} Given heterogeneous divisible resources, additive valuations, and an arbitrary Pareto optimal allocation $x$, we find resource prices and assign every agent a budget (in tokens) supports allocation $x$; i.e., where everyone is maximally happy given their tokens and resource prices, and all resources are fully allocated.
The equilibrium prices and token budgets can be computed in polynomial time. As tokens have no value, this is effectively in the linear Fisher market model.\footnote{Note that it is trivial to achieve {\sl some} Pareto optimal allocation supported by prices and budgets, for example set non-zero prices to all items, give one agent a budget that suffices to purchase all items, and give all other agents zero budget. The point is that we can support {\sl any} Pareto optimal allocation.}

\vspace{0.1in}
Consider the following toy example: \begin{itemize} \item 2 agents $i$ and $j$, 1 resource. Agent $i$ has value 1 for the entire resource, agent $j$ has value 99 for the entire resource. We've [arbitrarily] decided upon an allocation $x$ where agent  $i$ receives $99\%$ of the resource and agent $j$ receive $1\%$ of the resource, respectively. Note that both agents have the same utility from their allocation. 
	\item Set prices $p$ so that the resource has a price of 1 token, set budgets $b$ so agent $i$ gets budget $b_i=0.99$ tokens, and agent $j$ gets $b_j = 0.01$ tokens.  Now, given allocation $x$, both $i$ and $j$ are maximally happy with their allocations, subject to the price of the resource and their budget in tokens.
\end{itemize}
We call this a {\sl competitive equilibria with unequal budgets supporting allocation $x$}. We remark that the allocation $x$ is not envy-free as agent $j$ would be much happier with the $99\%$
of the resource assigned to agent $i$ than with the $1\%$ she has in hand, but --- given the number of tokens she has, she can do no better.


If our goal was to maximize social welfare we would give everything to the agent that values it most, in this case agent $j$. Arrow-Debreu competitive equilibrium prices can be any price $1 < p < 99$ for the resource. This will ensure that only agent $j$ will seek to purchase the resource and has reason to do so.

Rawls \cite{Rawls99} and his notion of justice is (arguably) to improve the lot of the weakest member of society, the so called Max-Min allocation\footnote{There are other notions of fair allocations also called max-min but are quite distinct from the Rawls notion, see \cite{Budish2011}.}. If we were to follow Varian's approach and give ``equal income" (in tokens) to both agents $i$ and $j$, the only possible CEEI is to split the item equally between the two agents (the price would be the total number of tokens assigned). This would indeed be envy free, but the utility to agent $i$ would be $1/2$ whereas the utility to agent $j$ would be $49.5$ --- not quite meeting Rawlsian justice.

In fact, in the toy example above we arbitrarily decided upon a partition of the item between agents $i$ and $j$. The partition is Pareto optimal because both agents desire the entire item. But in fact this specific partition also maximizes the Rawls max-min justice criteria. That said, any other partition would also be Pareto optimal. It now follows from our more general main result that any such allocation can be supported by appropriate prices and token budgets.

Recap: Walras, Fisher, and Arrow-Debreu seek a social welfare maximizing allocation supported by prices, Varian seeks fairness by using equal incomes. Contrawise, we start with an arbitrary Pareto optimal allocation, and show how to support this allocation with prices and token budgets. Our work can be viewed as a generalization of Walras and Fisher from the social welfare maximizing goal, to more general goals. One such goal is max-min allocations but one can consider weighted versions of such allocations and many other alternatives.
To the best of our knowledge, we are the first to give a competitive equilibrium for arbitrary Pareto optimal allocations. We achieve this result by allowing unequal budgets, while maintaining anonymous prices.

\subsection{Roadmap for this paper}

The model and problem are described in Section \ref{sec:model}. 

The proof of the main result appears in  Section \ref{sec:pricingbudgets}. The proof includes graph theoretic constructions as well as Brouwer fixed point arguments. A high level overview of the proof is provided in Section \ref{sec:overviewpricing} and the details appear subsequently.   

While the scenario assumed in Section \ref{sec:pricingbudgets} is that some arbitrary desirable Pareto optimal allocation is determined extraneously, in Section \ref{sec:maxmin} we specifically consider the [Pareto optimal] Rawsian max-min allocation and show that it can be computed efficiently. In special cases such as 2 agents or 2 items it becomes particularly efficient to do so. 

We conclude in Section \ref{sec:discussion} with a discussion and open problems.

\section{Model and Problem}
\label{sec:model}

There is a set $A$ of $n$ agents $\{1, \ldots , n\}$ and a set $I$ of $m$ fractionally divisible
items. Agent $i$ has value $v_{i,j}>0$ for item $j$, {\sl i.e.}, all valuations are strictly positive. 

An allocation $x= \{x_{i,j}\}$, $0\leq x_{i,j} \leq 1$, has the property that for any item
$j$ we have $\sum_{i=1}^n x_{i,j}\leq 1$. We also refer to the allocation for agent $i$, $x_i=x_{i,1}, \ldots, x_{i,m}$, where $x_{ij}$ is the fraction of item $j$ allocated to agent $i$.   

We consider additive valuations. I.e., the valuation of agent $i$ for an allocation $x_i$ is
$$v_i(x_i) = \sum_{j \in I} x_{i,j}v_{i,j}.$$

\subsection{Main problem} \label{The problem}

Given some Pareto optimal allocation $x$, our goal is to compute both prices and budgets so that $x$ is a competitive equilibrium of the resulting market.
That is, we seek to assign tokens to agents, and prices in tokens to items, so that every agent is maximally happy with her allocation $x_i$, based on the utility function specified below. 

Tokens are divisible. Every agent $i\in A$ has a budget $b_i$ (in tokens), and every item $j\in I$ has a price $p_j$ (in tokens).

\begin{definition} \label{def:utility} Given prices $p_j$, $j\in I$, and a budget of $b_i$ tokens, the utility to agent $i$ of an allocation
$x_i$, is $$u_i(x_i,p,b_i) = \left\{
\begin{array}{ll}
    \sum_{j=1}^m x_{i,j} v_{i,j}, & \hbox{when $\sum_{j\in I} x_{i,j} p_j \leq b_i$;} \\
    -1, & \hbox{otherwise.}
\end{array}
\right.$$
\end{definition}



Given a price vector $p$ and a budget vector $b$, an allocation $x_i$ is said to be in the demand set of agent $i$ if $u_i(x_i,p,b_i) = \max_{x'_i\in X_i} u_i(x'_i,p,b_i)$, where $X_i$ is the set of all possible allocations for agent $i$.

A tuple $(p,b)$ of price and budget vectors is said to support allocation $x$ if for every agent $i$, $x_i$ is in the demand set of $i$ given prices $p$ and budget $b_i$.

A tuple of prices $p$, budgets $b$, and allocation $x$ is a {\em competitive equilibrium} if the tuple $(p,b)$ supports the allocation $x$.




\section{Pricing and Budgets that Support Arbitrary Pareto Optimal allocations}
\label{sec:pricingbudgets}

\subsection{Overview} \label{sec:overviewpricing}

Given a Pareto optimal allocation $y$,  our goal is to establish prices $p$ and budgets $b$ that support this allocation in a competitive equilibrium.
We first modify the initial allocation $y$ to an allocation $x$ that has``less sharing'' of items between agents, while preserving the agent utilities. 
(Later, we will show that the prices and budgets supporting allocation $x$ also support the original allocation $y$.)

The proof proceeds as follows: 

We construct a bipartite graph $G(y)$ where vertices on the right correspond to agents and vertices on the left correspond to items. 
We add an edge between vertices $i$ and $j$ if agent $i$ is allocated a non-zero quantity of item $j$ in $y$.

Our first goal is to find an allocation $x$, such that every agent $i$ has the same utility as in allocation $y$, and the corresponding graph of $x$, $G(x)$, is acyclic (see Section~\ref{sec:cycle-elimination}). We give a constructive proof where we generate flow along a cycle (adjusting allocations) while maintaining agent utilities --- eventually the cycle is broken (see \Cref{prop:cycle-free-pareto}). We repeat this, and appropriately update the graph, until all cycles are removed. Let $x$ denote the resulting allocation.  
By construction, $G(x)$ is a forest. Note that every agent and every item belongs to exactly one tree. 

Our next step is to determine prices and budgets for items and agents within a single tree of the forest (see Section~\ref{sec:price-tree}). Conceptually, we set item prices such that every agent is indifferent between all of her adjacent items ({\sl i.e.}, items for which she has a non-zero allocation). We set some arbitrary agent to be the root of the tree and set prices to her adjacent items so that the agent will be indifferent between them. We now price the remaining tree going top down from the root. Every agent we encounter along the way has exactly one item that is already priced, and we price the remaining items adjacent to her so that she is indifferent (using her valuations). 

This shows that agents cannot benefit by shifting tokens from one adjacent item to another. We still need to argue that agents will not seek to deviate by shifting tokens to items for which they have zero allocation; this is proved in \Cref{lem:devtree}.

If the outcome of the cycle elimination process results in $G(x)$ being a single tree, we are done. It remains to consider the case where there are multiple trees, each of which has its own pricing (see Section~\ref{sec:price-forest}). 
To this end, we compute tree multipliers $\alpha_T$ for each tree $T$, which are used to scale both prices and budgets associated with a tree. Our goal is to find tree multipliers such that the resulting prices and budgets constitute a competitive equilibrium supporting the  allocation $x$. 
We first prove the existence of such multipliers, using Brouwer's fixed point theorem \cite{brouwer1911abbildung, brouwer1952intuitionist} (see \Cref{prop:alphas}).
Then, equipped with the existence of such multipliers, we give an efficient algorithm to compute them, using linear programming.

The last stage in our construction is to show that the prices and budgets that support the allocation $x$ also support the original Pareto optimal allocation $y$ (See \Cref{thm:origalloc}). 

\subsection{Graph representation and Cycle Elimination}
\label{sec:cycle-elimination}

Consider some arbitrary Pareto optimal allocation $y$. Recall that all valuations are strictly positive, therefore all items must be fully allocated (otherwise, $y$ is not Pareto optimal). However, it is possible that some agents receive an empty allocation.





We build a bipartite graph $G(y)$, where vertices on the right represent agents and vertices on the left represent items.
If $y_{i,j}>0$ we add an edge between agent $i$ and item $j$ with weight $y_{i,j}$.

We first give a graph theoretic framework for our problem.
Consider the following definition:

\begin{definition}\label{def:cycle-free}
	An allocation $x$ is {\em cycle-free} if the graph $G(x)$ is a forest.
\end{definition}

\begin{proposition}
	\label{prop:cycle-free-pareto}
	Given a Pareto optimal allocation $y$, there exists a cycle-free Pareto allocation $x$ that preserves the same agent utilities as in $y$. 
	Moreover, the allocation $x$ can be computed in polynomial time.
\end{proposition}

\begin{proof}
	Assume that $y$ is not cycle-free; We construct a forest graph $G(x)$ associated with another Pareto optimal allocation $x$, which is cycle-free.
	Let $C=i_1 \rightarrow j_1 \rightarrow i_2 \rightarrow j_2 \rightarrow \cdots \rightarrow j_k \rightarrow i_1$ be a simple cycle in $G(y)$, where $i_{\ell}$'s and $j_{\ell}$'s are agents and items, respectively.
	Each edge in the cycle $C$ represents an allocation $0 < y_{i,j} < 1$. $y_{i,j} > 0$ by definition of the graph construction, and $y_{i,j} < 1$ since the corresponding edge is part of a cycle, and therefore item $j$ cannot be allocated only to agent $i$.
	
	We define a new allocation $x$ as follows.
	For every edge $(i,j)$ not on the cycle $C$, $x_{i,j}=y_{i,j}$.
	Set $x_{i_1,j_1}=y_{i_1,j_1}+\epsilon_1$ (for a sufficiently small $\epsilon_1\geq0$, to be determined later).
	Set $x_{i_2,j_1}=y_{i_2,j_1}-\epsilon_1$.
	Next, set $\epsilon_2 = \left({\epsilon_1 v_{i_2,j_1}}\right)/\left({v_{i_2,j_2}}\right)$, and let $x_{i_2,j_2}=y_{i_2,j_2}+\epsilon_2$.
	One can easily verify that $\epsilon_2$ is such that agent $i_2$ is now indifferent between $y$ and $x$. In addition, since $v_{i_2,j_2} > 0$ the value of $\epsilon_2$ is well defined.
	We continue in this manner along the cycle, until reaching agent $i_1$, where $x_{i_1,j_k}=y_{i_1,j_k}-\epsilon_k$.
	
	All agents along the cycle are indifferent between $y$ and $x$ by construction, except, possibly, for agent $i_1$.
	If agent $i_1$ is made worse off by this process, then reverse the transitions along the cycle to make agent $i_1$ better off, without hurting any other agent, contradicting Pareto optimality.
	
	One can verify that there exists $\epsilon_1$ such that $x_{i,j}\geq 0$ for all $i,j$, and there exists an agent $i_{\ell}$ such that $x_{i_{\ell},j_{\ell-1}}=0$  (for $\ell=1$, $\ell-1=k$). Therefore, the cycle $C$ has been removed. The desired value of $\epsilon_1$ can be efficiently computed by testing all agents on the cycle as candidates for $\ell$.
	Continue in this manner, resetting $y=x$, until all cycles have been removed.
	
	Throughout the process, no agent is made worse off, therefore Pareto optimality is preserved throughout the process, and the corresponding graph of the obtained allocation is a forest, thus $x$ is cycle-free. Moreover, utilities are unchanged as desired.
	The process of creating a cycle-free allocation completes in polynomial time, since each iteration searches for a cycle in the underlying graph and then removes an edge from it. Therefore, the number of iterations is bounded by the number of edges is the graph.
\end{proof}


\subsection{Pricing a single tree}
\label{sec:price-tree}

Now, there are no cycles in $G(x)$ and we have a forest in hand.
For each tree separately, we price each item, and endow each agent with some \budget, such that no agent wishes to deviate to another affordable allocation within the tree.

\begin{definition}
	Given a set of items $I$ with item prices $\{p_j\}_{j \in I}$, the {\emph demand correspondence of agent $i$ with budget $b_i$}, is the set of all affordable fractional allocations of items in $I$ that maximize $u_i(x_i, p, b_i)$.
\end{definition}

Given a tree $T$ in $G(x)$, we focus on the case that it contains vertices representing at least one item and one agent\footnote{Other cases are uninteresting. If a tree contains two or more vertices then it must contain both an item and an agent. If a tree contains only a single vertex, then if this single vertex is an agent then this agent has an empty allocation --- in this case we set the agent's budget to zero. The single vertex cannot be an item,  since all item valuations are strictly greater than zero, hence, any Pareto optimal allocation must allocate all items.}. We price the items in the tree $T$ successively item by item as follows: 
\begin{enumerate}
	\item If no item has been priced yet, choose an arbitrary agent $i \in T$. For all adjacent items $j$ set the price of item $j$, $p_j=v_{i,j}$.
	\item Repeat the following until all items in $T$ have been priced:
	
	Choose some arbitrary agent $i'\in T$ that is adjacent to some item already priced, say item $j'$, but also adjacent to at least one unpriced item. As $T$ is a tree, and by the construction of the pricing mechanism, $i'$ is adjacent to at most one item that has been priced.
	Agent $i'$ should be indifferent between all items agent $i'$ is adjacent to, so the price (in tokens) for any two items, $k$ and $k'$, both adjacent to agent $i'$ should satisfy
	\begin{equation}
		p_k/p_{k'} = v_{i',k}/v_{i',k'}. \label{eq:indifference} \end{equation}
	
		Ergo, to ensure that agent $i'$ is indifferent to spending tokens amongst neighboring items, for every item $k \neq j'$ such that $x_{i',k} > 0$, set the price for item $k$: $$p_k = v_{i',k} \cdot p_{j'} / v_{i',j'}.$$

	\item Set the budget for every agent $i$ so that she can afford exactly the appropriate fractions of her adjacent items, i.e., $b_i = \sum_{j}x_{i,j}p_j$.
\end{enumerate}

We call a tree where all items have been priced a ``{\sl priced tree}''.

The following lemma considers what happens when an agent changes the allocation of tokens between adjacent items in a Pareto optimal allocation.

\begin{lemma}
	Given a priced tree $T$ in $G(x)$, consider an agent $i$ and items $j,j'$ such that $0<x_{i,j}<1$, $0<x_{i,j'}\leq 1$.
	Let $$x'_{i,j}=x_{i,j}+\epsilon{p_{j'}}/{p_{j}}, \mbox{\rm\ and\ }x'_{i,j'}=x_{i,j'}-\epsilon, \mbox{\rm for  a sufficiently small $\epsilon$}.$$
	Then, $u_i(x)=u_i(x')$, and the payment of agent $i$ for allocations $x$ and $x'$ is identical.
	\label{lem:internal}
\end{lemma}

\begin{proof}
	Let $\epsilon \leq \min\left\{x_{i,j'},(1-x_{i,j}){p_{j}}/{p_{j'}}\right\}$.
	The payment of agent $i$ changes by $\epsilon (p_{j'}/p_j) \cdot p_{j} - \epsilon p_{j'} = 0$.
	The utility of agent $i$ changes by $\epsilon (p_{j'}/{p_j}) \cdot v_{i,j} - \epsilon v_{i,j'} = 0$,
	where the last equation follows by ${v_{i,j}}/{p_j} = {v_{i,j'}}/{p_{j'}}$, which holds due to Equation \eqref{eq:indifference}.
\end{proof}

The following lemma establishes a sufficient requirement that makes deviating from the proposed allocation beneficial to an agent.

\begin{lemma}
	Given a priced tree $T$ in $G(x)$, consider an agent $i$ and items $j,j'$ such that $x_{i,j}>0$, $x_{i,j'}=0$.
	Let $$x'_{i,j}=x_{i,j}-\epsilon{p_{j'}}/{p_j}, \mbox{\ and\ } x'_{i,j'}=\epsilon, \mbox{\rm\ for a sufficiently small $\epsilon$}.$$
	If ${v_{i,j'}}/{p_{j'}}>{v_{i,j}}/{p_{j}}$, then $u_i(x_i, p, b_i) < u_i(x'_i, p, b_i),$ and the payment of agent $i$ for $x$ and $x'$ is identical.
	\label{lem:deviator}
\end{lemma}

\begin{proof}
	Let $\epsilon \leq \min\left\{1,x_{i,j}{p_{j}}/{p_{j'}}\right\}$.
	The payment of agent $i$ changes by $\epsilon p_{j'} - \epsilon {p_{j'}}/{p_{j}} \cdot p_{j} = 0$.
	The utility of agent $i$ changes by $\epsilon v_{i,j'} -\epsilon {p_{j'}}/{p_j} \cdot v_{i,j} > 0$,
	where the inequality follows by our assumption that $\frac{v_{i,j'}}{p_{j'}}>\frac{v_{i,j}}{p_{j}}$.
\end{proof}

We now show that no agent seeks to reallocate tokens from an adjacent item to a non-adjacent item in the same tree. 

\begin{lemma}
	\label{lem:devtree}
	Given a priced tree $T$ in $G(x)$, let $i$ be an arbitrary agent in $T$ and let $j,j'$ be two items in $T$ such that $x_{i,j}>0$ and $x_{i,j'}=0$. 
	Then,
	$$
	\frac{v_{i,j}}{p_{j}} \geq \frac{v_{i,j'}}{p_{j'}}.
	$$
\end{lemma}

\begin{proof}
	Assume for contradiction that there exist $i,j,j'$ such that $\frac{v_{i,j}}{p_{j}} < \frac{v_{i,j'}}{p_{j'}}$.
	Then, consider the cycle obtained by the path connecting $i$ to $j'$ in $T(x)$ and the additional (undirected) edge from $j'$ to $i$. Note that $j$ is not necessarily in the cycle. However, the item adjacent to $i$ in the cycle was priced relatively to $j$.
	We construct a new allocation $x'$ such that agent $i$ strictly improves her utility and all other agents along the path are indifferent, contradicting the Pareto optimality of $x$.
	
	Fix some sufficiently small $\epsilon$.
	Rename the agents along the cycle $1, \ldots, k$, where $i=k$.
	Similarly, rename the items along the cycle so that item $j_\ell$ is shared by agents $\ell$ and $\ell+1$ along the cycle for $\ell=1,\ldots,k-1$ (and item $j_k=j_0$ is shared between agent $k$ and $1$).
	
	Let $x'_{1,k} = x_{1,k}-\epsilon$.
	For every $\ell=1, \ldots, k-1$, let \begin{eqnarray*}x'_{\ell,j_{\ell-1}} &=& x_{\ell,j_{\ell-1}} - \epsilon\frac{p_{k}}{p_{\ell-1}}\\ x'_{\ell,j_{\ell}} &=& x_{\ell,j_{\ell}} + \epsilon  \frac{p_{k}}{p_\ell}.\end{eqnarray*} Also, let \begin{eqnarray*} x'_{k,j_{k-1}} &=& x_{k,j_{k-1}} - \epsilon \frac{p_{k}}{p_{k-1}} \\ x'_{k,j_{k}} &=& \epsilon.\end{eqnarray*}
	
	We claim that agents $\ell=1,\ldots,k-1$ are indifferent, by Lemma~\ref{lem:internal}. This is clear for agent $1$, and for agents $\ell=2, \ldots, k-1$, note that $\epsilon \frac{p_{k}}{p_{\ell-1}} \cdot \frac{p_{\ell-1}}{p_{\ell}} = \epsilon\frac{p_{k}}{p_{\ell}}$.
	
	Finally, agent $k$ is better off by Lemma~\ref{lem:deviator}, which contradicts the Pareto optimality of $x$.
\end{proof}


\subsection{Pricing the Complete Allocation (a Forest)}
\label{sec:price-forest}

We seek anonymous prices for all items, guaranteeing that no agent seeks to deviate. We know that no agent within some tree $T$ seeks to change her allocation so as to gain some other item, not allocated to her, within the tree $T$.

We now need to deal with possible deviations across trees, {\sl i.e.}, some agent in tree $T$ wishes to use some of her tokens so as to get an item that resides in tree $T'\neq T$. We show that there exist scalar multipliers, for every tree $T$, such that multiplying the prices and budgets computed above by this scalar value suffices to ensure that there are no profitable deviations across trees. Such multipliers are called {\sl good multipliers}.

We remark that $b_i$ and $p_j$ represent budgets and prices as computed individually tree by tree. 

We consider scaling the prices and budgets of each tree $T_i$ by a
constant $\alpha_i>0$. The scaling is normalized by assuming that
$\sum_i \alpha_i=1$. Let $\Delta$ be the simplex, i.e.,
$\Delta=\left\{\alpha\in \mathbb{R}^k: \alpha_i\geq 0, \sum_{i=1}^k \alpha_i=1\right\}$.

Note that the scaling does not change the incentives of the agents within a particular tree, but it may change the incentives between tress. Our goal is to find scalar multipliers such that no agent in some tree $T$ will seek to deviate and get an item from another tree $T'$. 

If the forest $G(x)$ contains degenerate single vertex trees representing agents receiving an empty allocation, we omit them from the scaling process. Recall that such an agent has a budget of zero, and therefore can only purchase items if they are free. The pricing scheme of each tree separately does not produce free items, but this may occur if we choose to scale the prices by a factor of zero. Later on, we will observe that this undesired outcome does not happen.

\subsubsection{Proving that good multipliers exist.}
\label{sec:fixed-existence}

\begin{proposition}
	\label{prop:alphas}
	There exist scaler multipliers $\alpha_T$, for every tree $T$, such that scaling all prices and budgets from the tree by $\alpha_T$ results in that either:
	(1) no agent can gain by deviating, or (2) for each tree $T_i$ there is some other tree $T_j$ such that there is an agent in tree $T_j$ that benefits by deviating and buying an item in tree $T_i$.
\end{proposition}

We will later show that if condition (2) holds, then this contradicts the Pareto optimality of the initial allocation $y$ (and $x$).

We use Brouwer's fixed-point theorem to prove the proposition.
Fix an agent $i$ and an item $j$.
Let $v_{i,j}$ denote agent $i$'s valuation for item $j$.
Let $p_j$ denote the price of item $j$, and $b_i$ denote the budget of agent $i$.
Finally, let $u_i=\sum_j x_{ij} v_{ij}>0$ be the utility of agent $i$, given the current allocation $\{x_{ij}\}$.

Let $T(i)$ be the tree that has agent $i$, and $T(j)$ be the tree that has item
$j$. Thus, prices in $T(i)$ are multiplied by $\alpha_{T(i)}$ ($\widehat{p}_i \triangleq \alpha_{T(i)} p_i$) and prices in $T(j)$ are multiplied by $\alpha_{T(j)}$ ($\widehat{p}_j \triangleq \alpha_{T(j)} p_j$).  Let $v_{\min}=0.5 \min_{i,j | v_{i,j}>0} v_{i,j}$. \footnote{The constant $0.5$ can be replaced by any constant $c \in (0,1)$; it is used to ensure that $GAIN_{T}$ is a continuous function.}

We define a variable $GAIN_{i,j}$ defined via Algorithm \ref{alg:vertex-arrival}, where $i$ is an agent and $j$ is an item, in different trees. $GAIN_{i,j}$ represents the profit that agent $i$ can obtain by shifting tokens from the current allocation to item $j$, capped by some small value. Our eventual goal is to prove that there exist multipliers such that $GAIN_{i,j}=0$ for all $i$, $j$, which implies that no agent wishes to purchase an item  from another tree. 

We know that an agent is indifferent between all items for which she has a positive allocation. To release one token (to be used to purchase some of item $j$) agent $i$ has to give up utility $u_i/(\alpha_{T(i)} b_i)$ (recall that $u_i$ is the current utility to agent $i$).

In the computation of $GAIN_{i,j}$ we consider three cases: (a) agent $i$ can afford to purchase all of item $j$ (lines \ref{line:ca1}, \ref{line:ca2} in the algorithm), (b) agent $i$ gets item $j$ for free (lines \ref{line:cb1}, \ref{line:cb2} in the algorithm) and (c) agent $i$ can afford to purchase only part of item $j$ (line \ref{line:cc1} in the algorithm).

If indeed agent $i$ will have positive profit,  then in case (a) all the item will be purchased and she will give up some of the current allocation, in case (b) she will simply add the item to her allocation for free, and in case (c) the agent will give up her entire current allocation and add as much of item $j$ as possible.

\begin{algorithm}
	\caption{Calculating $GAIN_{i,j}$}
	\label{alg:vertex-arrival}
	\begin{algorithmic}[1]
		\STATE {\bf Inputs:} {Agent $i\in A$, item $j\in I$, that reside in distinct trees: $T(i)\neq T(j)$.}
		\STATE {\bf Output:} {$GAIN_{i,j}\in [0,v_{\min}]$, where $v_{\min}=0.5 \min_{i,j} v_{i,j}$}
		\IF{$\alpha_{T(j)} p_j < \alpha_{T(i)}b_i$} \label{line:ca1}
		\STATE $GAIN_{i,j} = \max\left(v_{i,j}-\alpha_{T(j)}p_j \frac{u_i}{\alpha_{T(i)}b_i},0\right)$ \label{line:ca2} 
		\ELSE 
		\IF{$\alpha_{T(j)}=0$} \label{line:cb1}
		\STATE $GAIN_{i,j}=v_{i,j}$ \label{line:cb2}
		\ELSE
		\STATE \label{line:cc1} $GAIN_{i,j}=\max\left(v_{i,j}\frac{\alpha_{T(i)}b_i}{\alpha_{T(j)}p_j}-u_i,0\right)$
		\ENDIF
		\ENDIF
		\STATE \Return $GAIN_{i,j}=\min\left(GAIN_{i,j},v_{\min}\right)$
	\end{algorithmic}
\end{algorithm}


Given variables $GAIN_{i,j}$ for every agent $i\in A$ and item $j\in I$, let $GAIN_j=\max_i GAIN_{i,j}$, and for every tree $T$, let $$GAIN_T=\max_{\mbox{\rm items\ }j\in T} GAIN_j.$$ 

Ergo, $GAIN_{i,j}$ denotes the utility agent $i$ can gain by shifting her [adjusted] budget toward item $j$, but this is capped by some arbitrary (small) value $v_{\min}$. $GAIN_j$ denotes the maximal potential (capped) gain for shifting all the [adjusted] budget to item $j$,  amongst all possible agents $i$ should they divert to purchasing item $j$, and $GAIN_T$ denotes the maximal potential (capped) gain amongst all agents should these agents divert to some item in $T$.

We now define a mapping from $\alpha\in \Delta$ to $F(\alpha)$ as follows:
\[
F(\alpha)_T = \frac{\alpha_T+GAIN_T}{1+\sum_{T'} GAIN_{T'}}
\]

The conditions required to apply Brouwer's fixed point theorem are ($i$) $\Delta$ is convex and closed ($ii$) $F(\alpha) \in \Delta$ for all $\alpha\in \Delta$ and ($iii$) $F(\alpha)$ is continuous. 

Clearly the simplex $\Delta$ is a convex closed set. We show the following:
\begin{lemma}
	\label{cl:continuous}
	The mapping $F$ satisfies: (1) If $\alpha\in \Delta$, then
	$F(\alpha)\in \Delta$, (2) $F$ is a continuous mapping over the
	simplex $\Delta$.
\end{lemma}

\begin{proof}
	We first prove (1). Clearly $F(\alpha)\geq 0$. It remains to show that $\sum_T F(\alpha)_T=1$. Indeed,
	\[
	\sum_T F(\alpha)_T = \sum_T \frac{\alpha_T+GAIN_T}{1+\sum_{T'}
		GAIN_{T'}} =  \frac{\sum_T \alpha_T+ \sum_TGAIN_T}{1+\sum_{T'}
		GAIN_{T'}}=1,
	\]
	where the last equality follows by $\sum_T \alpha_T =1$ (since $\alpha\in \Delta$).
	Hence $F(\alpha)\in \Delta$.
	
	We next prove (2). If $\alpha_T >0$ for all $T$, then $F$ is clearly
	continuous. Otherwise, there exists a subset of trees such that $\alpha_T=0$ for every $T$ in this subset. But there must exist a tree $T$ for which $\alpha_T>0$ (since $\sum_T \alpha_T =1$).  For these values of $\alpha$, we will show that for every $j$, $GAIN_j$ is continuous, inferring that $GAIN_{T(j)}$ is also continuous and therefore so is $F(\alpha)$.
	
	If $\alpha_{T(i)},\alpha_{T(j)}>0$, then clearly $GAIN_{i,j}$ is continuous in $\alpha$, and so is $GAIN_j$.
	
	If $\alpha_{T(i)}=0$ and $\alpha_{T(j)}>0$, then
	$$GAIN_{i,j}=\min(v_{\min},\max(v_{i,j}\frac{\alpha_{T(i)}b_i}{\alpha_{T(j)}p_j}-u_i,0)),$$
	and $GAIN_{i,j}$ is continuous --- hence $GAIN_j$ is continuous, and so is $GAIN_{T(j)}$. 
	Specifically, for $\alpha_{T(i)}=0$ we have that 	$GAIN_{i,j}=\min(v_{\min},\max(0-u_i,0))=0$ since $u_i>0$. When we change  $\alpha_{T(i)}=\epsilon$, we still have $GAIN_{i,j}=0$, for a sufficiently small $\epsilon>0$, {\sl i.e.}, for $\epsilon\leq \frac{u_i p_j \alpha_{T(j)}}{b_i v_{ij}}.$
	
	The only case not previously considered is when $\alpha_{T(j)}=0$. 
	In this case, regardless of the value of $GAIN_{i,j}\in[0,v_{\min}]$ we will show that $GAIN_j=v_{\min}$.	
This follows since there
	exists a tree $T$ such that $\alpha_T>0$ and all agents $k$ in $T$ have $GAIN_{k,j}=v_{\min}$, which implies that $GAIN_{j}=v_{\min}$.
	When we have  $\alpha_{T(j)}=\epsilon$, for a sufficiently small $\epsilon>0$, we claim that all the agents $k$ in $T$ still have  $GAIN_{k,j}=v_{\min}$. This follows since $v_{k,j}\frac{\alpha_{T(k)}b_k}{\alpha_{T(j)}p_j}$ can be made arbitrarily large --- as $\alpha_{T(k)} > 0$  and $b_k > 0$. This follows since we've temporarily omitted agents with zero budgets from our analysis (recall that for the fixed-point and scaling we ignore degenerate trees that represent agents with zero budgets and empty allocations).
\end{proof}

It is worth noting that $GAIN_{i,j}$ is not a continuous function.\footnote{
Assume that $\alpha_{T(i)}=\beta\alpha_{T(j)}$ and let both converge to
zero. There exist ranges of values of $\beta$ and $u_i$ such that $GAIN_{i,j}= \max( v_{i,j}
(b_i/p_j)\beta-u_i,0)$ is discontinuous. This holds since different $\beta$ values give different $GAIN_{i,j}$ values, even for infinitesimally close values of $\alpha_{T(i)}$ and $\alpha_{T(j)}$.
}

We get around this by analyzing $GAIN_j$, which is continuous. This is achieved since when the multiplier $\alpha_{T(j)}$ is negligible, then agents from other trees would have a significant gain from deviating and purchasing  item $j$ (since its cost is negligible).

Given \Cref{cl:continuous}, we get the following as a direct corollary of Brouwer's fixed-point theorem.
\begin{corollary}
	There exists an $\alpha\in \Delta$ such that $F(\alpha)=\alpha$.
\end{corollary}

We claim that at the fixed-point, $\alpha_T\neq 0$ for every tree $T$.

\begin{lemma}
	\label{lemma:Brouwer-zero}
	For $\alpha$ such that $F(\alpha)=\alpha$, it holds that $\alpha_T\neq 0$ for every tree $T$.
\end{lemma}

\begin{proof}
	At the fixed point, for every tree $T$, we have:
	\[
	\alpha_T = \frac{\alpha_T+GAIN_T}{1+\sum_{T'} GAIN_{T'}}
	\]
	This implies that $\alpha_T \sum_{T'} GAIN_{T'}=GAIN_T$.
	
	Assume for contradiction that $\alpha_T=0$. It follows that $GAIN_T=0$.
	But since $\sum_{T'} \alpha_{T'}=1$, there exists a tree $T'$
	such that $\alpha_{T'}>0$. Since all valuations are strictly
	positive, all the agents in tree $T'$ would prefer to buy any
	item in tree $T$ (since the price is zero and it does not affect their
	budget). This implies that $GAIN_T=v_{\min}$, which is a
	contradiction.
\end{proof}

We proceed with the following lemma.

\begin{lemma}
	Let $\alpha$ be such that $F(\alpha)=\alpha$. Then either (1) $GAIN_T>0$ for every tree $T$, or (2) $GAIN_T=0$ for every tree $T$.
	\label{claim:Brouwer-cycle}
\end{lemma}


\begin{proof}
	Assume towards contradiction that there exist $T,T'$ such that $GAIN_T=0$ and $GAIN_{T'}>0$. Since
	$\alpha$ is a fixed-point, it holds that
	\begin{equation}
	\label{eq:fixed}
	\alpha_T = \frac{\alpha_T+GAIN_T}{1+\sum_{T''} GAIN_{T''}}.
	\end{equation}
	Moreover, by Lemma \ref{lemma:Brouwer-zero}, $\alpha_T>0$.
	Since $GAIN_T=0$, the numerator of the RHS of \eqref{eq:fixed} equals $\alpha_T$. 
	In addition, since $GAIN_{T'}>0$, we have $\sum_{T''} GAIN_{T''}>0$.
	Thus, the denominator of the RHS of \eqref{eq:fixed} is strictly greater than 1. It follows that the ratio is is strictly smaller than $\alpha_T$, contradicting \eqref{eq:fixed}. 
\end{proof}

\subsubsection{No Inter-Tree Deviations}
\label{sec:no-cycle}

In this section we show that there are no inter-tree deviations given our pricing and budgets. We first define a deviation graph over the graph $G(x)$.

\begin{definition}
	Let $G(x)$ be a forest induced from a cycle-free Pareto optimal allocation $x$. Let $p$ be some pricing of the items and $\alpha$ be some scaling of the prices in each tree. For each agent $i$ that is strictly better off purchasing an item $j$, add a directed edge $(i,j)$. The resulting graph is denoted the deviation graph of $(G(x), p, \alpha)$.
\end{definition}

Note that although the graph $G(x)$ is undirected, the deviation graph adds directed edges. Recall that by \Cref{claim:Brouwer-cycle}, for the fixed-point $\alpha$, either the deviation graph contains no directed edges at all, or it contains at least one directed edge into some item vertex within each tree.
By Definition \ref{def:cycle-free} $x$ is cycle-free. 


Given a deviation graph $(G(x), p, \alpha)$, consider a simplified graph $G(V, E)$ where every vertex in $V$ corresponds to a tree in $G(x)$, 
and we add a directed edge between vertices in $V$ iff there is at least one directed deviation edge between vertices in the corresponding trees in $G(x)$. Note that a vertex in $V$ corresponding to a tree $T$ in $G(x)$ has incoming edge(s) if and only if $GAIN_T > 0$. 

The following is a well known fact; we include a proof for completeness.
\begin{lemma}
	\label{lem:gt-cycle}
	For any directed graph $G(V,E)$, if for each vertex the in-degree is at least 1, there is a directed cycle in $G$.
\end{lemma}

\begin{proof}
	Start from some arbitrary vertex $v_1$. Since $v_1$'s in-degree is at least $1$, there exists a directed edge $(v_2, v_1)$. Move to $v_2$ and repeat the backward walk on the graph. Since the graph is finite, the path must reach a previously visited vertex. The reverse order of the visited vertices contains a cycle. 
\end{proof}

By Proposition~\ref{prop:alphas} and Lemma~\ref{lem:gt-cycle}, if some agent wishes to diverge, the simplified graph corresponding to $G(x)$ contains a directed cycle.

We proceed as follows:

\begin{itemize}
	 
	\item This directed cycle can be expanded to an agent-item cycle, such that when the agent and item are in the same tree, the allocation is non zero.
	\item In this case, we show that the allocation can be improved by changing the allocation along the cycle, while contradicting Pareto optimality. 
	\item In summary, it follows that there cannot be an agent in one tree who wishes to deviate to an item in another tree. 
\end{itemize}


We extend the directed cycle in the simplified graph $G(V, E)$ to a directed cycle in the deviation graph $(G(x), p, \alpha)$. To construct a directed cycle one can use any undirected edge of $G(x)$, either from an item to an agent or from an agent to an item (but not both).

A cycle in $G(x)$ is an {\em alternating cycle} if the vertices along the cycle alternate between vertices corresponding to agents and items.

\begin{lemma}
	Given a directed cycle in $G(V,E)$, 
	one can extend the cycle to an alternating cycle in $G(x)$, such that 
	for every edge along the cycle that connects an agent $i$ and an item $j$ that belong to the same tree, 
	the allocation $x_{i,j}$ is non zero. For agent-item pairs crossing trees, the agent has a strict preference to deviate.
\end{lemma}

\begin{proof}
	Fix a cycle between trees. Every tree $T$ along the cycle has an incoming edge to an item $j$ in $T$ and an outgoing edge to the next tree from an agent $i$ in $T$.
	There is a unique (undirected) path in $T$ from $j$ to $i$, representing non-zero allocations of items to agents along the path. By connecting all these paths by the directed edges between the trees, we obtain a cycle in the graph.
\end{proof}


\begin{definition}
	\label{def:adjusted}
For any item $j$, let $\widehat{p}_j = p_j \cdot \alpha_{T(j)}$, where $T(j)$ is the tree containing  item $j$. We refer to  $\widehat{p}_j$ as the adjusted price of item $j$. Likewise, let $\widehat{b}=b_i \cdot \alpha_{T(i)}$ be the adjusted budget of agent $i$, where $T(i)$ is the tree to which $i$ belongs. 
\end{definition}

By using adjusted prices $\widehat{p}$ instead of $p$, we can simplify our notation and remove $\alpha$ when it is fixed. We also use the adjusted budget $\widehat{b}$ in subsequent proofs. 

The following three lemmata are analogous to Lemmata \ref{lem:internal}, \ref{lem:deviator} and \ref{lem:devtree} used in the proof of no deviation within a single tree.

\begin{lemma}
	\label{lem:same-tree-cycle}
	Consider an agent $i$ and items $j,j'$ such that $0<x_{i,j}<1$, $0<x_{i,j'}\leq 1$.
	Let $x'_{i,j}=x_{i,j}+\epsilon{\widehat{p}_{j'}}/{\widehat{p}_{j}}$, and $x'_{i,j'}=x_{i,j'}-\epsilon$, for a sufficiently small $\epsilon$.
	Then, $u_i(x)=u_i(x')$, and the payment of agent $i$ for $x$ and $x'$ is identical.
	
\end{lemma}

\begin{proof}
	Since both $x_{i,j}$ and $x_{i,j'}$ are strictly positive, $T(i) = T(j) = T(j')$.
	Let $\epsilon \leq \min\{x_{i,j'},(1-x_{i,j}){\widehat{p}_{j'}}/{\widehat{p}_j}\}$.
	The payment of agent $i$ changes by $\epsilon {\widehat{p}_{j'}}/{\widehat{p}_j} \cdot \widehat{p}_{j} - \epsilon \widehat{p}_{j'} = 0$.
	The utility of agent $i$ changes by $\epsilon {\widehat{p}_{j'}}/{\widehat{p}_j} \cdot v_{i,j} - \epsilon v_{i,j'} = 0$,
	where the last equation follows by ${v_{i,j}}/{\widehat{p}_j} = {v_{i,j'}}/{\widehat{p}_{j'}}$, which holds due to Equation \eqref{eq:indifference}.
\end{proof}

\begin{lemma}
	\label{lem:diff-tree-cycle}
	Consider an agent $i$ and items $j,j'$ such that $x_{i,j}>0$, $x_{i,j'}=0$.
	Let $x'_{i,j}=x_{i,j}-\epsilon{\widehat{p}_{j'}}/{\widehat{p}_j}$, and $x'_{i,j'}=\epsilon$, for $\epsilon \leq \min\{1,x_{i,j}{\widehat{p}_{j}}/{\widehat{p}_{j'}}\}$.
	If ${v_{i,j'}}/{\widehat{p}_{j'}}>{v_{i,j}}/{\widehat{p}_{j}}$, then $u_i\left(x_i, \widehat{p}, \widehat{b}_i\right) < u_i\left(x'_i, \widehat{p}, \widehat{b}_i\right)$, and the payment of agent $i$ for $x$ and $x'$ is identical.
	
\end{lemma}

\begin{proof}
	The payment of agent $i$ changes by $\epsilon \widehat{p}_{j'} - \epsilon {\widehat{p}_{j'}}/{\widehat{p}_{j}} \cdot \widehat{p}_{j} = 0$.
	The utility of agent $i$ changes by $\epsilon v_{i,j'} -\epsilon {\widehat{p}_{j'}}/{\widehat{p}_j} \cdot v_{i,j} > 0$,
	where the last inequality follows by our assumption that ${v_{i,j'}}/{\widehat{p}_{j'}}>{v_{i,j}}/{\widehat{p}_{j}}$.
\end{proof}

\begin{lemma}
	\label{lem:paretodominance}
	Consider a deviation cycle $(i_1, j_1, i_2, j_2, \ldots, i_{k}, j_{k}, i_{k+1} = i_1)$,
	where $i_1$ and $j_1$ belong to different trees.
	Every agent $i_s$ (on the cycle) has associated $\epsilon_{s-1}$ and $\epsilon_s$ such that the allocation $x'_{i_s, j_s} = x_{i_s, j_s} + \epsilon_s$ and $x'_{i_s, j_{s-1}} = x_{i_s, j_{s-1}} - \epsilon_{s-1}$ is feasible. 
	Moreover, $x'$ Pareto dominates $x$. 
\end{lemma}

\begin{proof}
	Fix some sufficiently small $\epsilon_1$, to be determined later, and let $\epsilon_2 = \epsilon_1 {\widehat{p}_1}/{\widehat{p}_2}$. In general, $$\epsilon_s = \epsilon_{s-1} {\widehat{p}_{s-1}}/{\widehat{p}_s} = \epsilon_1 {\widehat{p}_1}/{\widehat{p}_s}.$$
	Note that $\epsilon_{k+1} = \epsilon_1{\widehat{p}_1}/{\widehat{p}_{k+1}} = \epsilon_1$ since $\widehat{p}_{k+1} = \widehat{p}_1$.
	
		We remark that as $i_1$ and $j_1$ belong to different trees, and there are no trivial trees (agents allocated nothing) it follows that all inter-tree directed edges go from an odd indexed agent $i$ to an odd indexed item $j$. In allocation $x$ nothing is allocated across different trees, {\sl i.e.}, the allocation of item $j$ to agent $i$ is zero. Note, however, that such allocations only increase and thus the resulting allocation is legitimate for sufficiently small values of the $\epsilon_1$.
	
	For every item $j_s$ in the cycle, the net change in its allocation is zero, because $x_{i_{s-1},j_s} + x_{{i_s},j_s} = x'_{i_{s-1},j_s} + x'_{{i_s},j_s}$.
	
	For every agent $i_s$, the change in the payment is zero, because $\epsilon_{s-1}\widehat{p}_{s-1} = \epsilon_1 = \epsilon_s \widehat{p}_s$.
	
	For every $s$, $x_{{i_s},j_{s-1}} > 0$ since $T(i_s) = T(j_{s-1})$, therefore, for $\epsilon_{s-1} \leq x_{{i_s},j_{s-1}}$ we have that $x'_{{i_s},j_{s-1}} \geq 0$.
	
	Since the net change in the allocation is zero and $x'_{{i_s},j_{s-1}} \geq 0$ it implies that $x'_{{i_s},j_{s}} \leq 1$.
	
	It remains to show that $x'$ Pareto dominates $x$. By \Cref{lem:same-tree-cycle} agents whose two adjacent items (along the cycle)
	are in the same tree are indifferent. By \Cref{lem:diff-tree-cycle}, an agent whose subsequent item resides in a different tree is strictly better off. This concludes the proof.
\end{proof}


Using the arguments above, we conclude that whenever payments and budgets are adjusted by multiplication by the $\alpha$ coordinates ($\alpha$ is a fixed point), it holds that $GAIN_T=0$ for all trees $T$; i.e., there are no deviations.  

\subsubsection{Computing the fixed point by a linear program}

Given an allocation $x$ where $G(x)$ is a forest, we now show how to compute a fixed point $\alpha$, using 
a linear program. 
Each agent $i$ is located in some tree $T(i)$, and each item $j$ is located in some (possibly different) tree $T(j)$.
The following set of constraints defines the set of fixed points.
By \Cref{sec:fixed-existence}, such a fixed point $\alpha$ exists. By \Cref{sec:no-cycle}, the fixed point supports allocation $x$. Recall the variables from definition \ref{def:utility}. 

The total ``bang per token" of agent $i$ is $u_i/\left(b_i \alpha_{T(i)}\right)$, {\sl i.e.}, to gain one token agent $i$ needs to give up $u_i/\left(b_i \alpha_{T(i)}\right)$ units of utility. 
If agent $i$ deviates and chooses to purchase some of item $j$, the profit per token is $v_{ij}/\left(p_j \alpha_{T(j)}\right)$.
The following linear program gives us the required scaling factors $\alpha_T$:  

\begin{eqnarray*}
	&&\max \lambda \\ \mbox{such that} \\
	\forall i \in A, j \in I, &&\frac{u_i}{b_i}\cdot {\alpha_{T(j)}} \geq \frac{v_{ij}}{p_j}\cdot {\alpha_{T(i)}}\\
	&&\sum_T \alpha_{T} = 1\\
	&&\lambda \leq \alpha_{T} \leq 1\end{eqnarray*}


By \Cref{lemma:Brouwer-zero} we are guaranteed that $\lambda>0$.

\noindent{\bf Computational complexity:} The complexity of pricing items and computing agent budgets is linear in $G(x)$. Finding the scaling factors can be done using the linear program above. In some special cases we can compute a desired allocation, the prices and budgets in nearly linear time (see Section \ref{sec:maxmin}).

Recall that the fixed point was computed ignoring degenerate trees that represent agents with zero budgets and empty allocations. It is now safe to reconsider these agents again in our analysis. Since for every non-degenerate tree $T$ we have in the fixed point $\alpha_T > 0$, there are no free items, and therefore these agents cannot afford any item. They also were not allocated any items.

To summarize: 

\begin{proposition}
	For any cycle-free Pareto optimal allocation there exist item prices and agent budgets supporting this allocation. Such prices and budgets can be computed efficiently.
\end{proposition}

\subsubsection{Supporting the Original Pareto optimal Allocation}


Recall that starting with an arbitrary Pareto optimal allocation $y$, we obtained a cycle-free Pareto optimal allocation $x$ with the same agent utilities as in $y$.
In this section we show that the prices and budgets that support allocation $x$ support the original allocation $y$ as well.

Recall Definition \ref{def:adjusted} where we use $\widehat{p}_j$ to denote the adjusted  price of item $j$ and $\widehat{b}_i$ to denote the adjusted budget of agent $i$. 


\begin{lemma}
	\label{thm:origalloc}
	Given an initial Pareto-optimal allocation $y$, the item prices and agent budgets obtained above to support the modified cycle-free allocation $x$ also support the original allocation $y$.
\end{lemma}

\begin{proof}
	Let $y$ be an arbitrary Pareto-optimal allocation, and let $x$ be the cycle-free Pareto-optimal allocation ensured by Proposition~\ref{prop:cycle-free-pareto}, which has the same agent utilities as in $y$. 
	
	The total payment of agent $i$ in $x$ is $\widehat{b}_i$, as every agent uses her entire budget, by construction. 
	If agent $i$ were to purchase $y_i$ using the prices $\{\widehat{p}_j\}$, her total payment would be $\sum_{j=1}^{m}y_{i,j}\widehat{p}_j$.
	We now show that this sum is also equal to agent $i$'s budget, $\widehat{b}_i$.
	Since all budgets are fully utilized in allocation $x$, and in any Pareto-optimal allocation all items are fully allocated, it follows that
	$$\sum_{i=1}^{n}\widehat{b}_i = \sum_{i=1}^{n} \sum_{j=1}^{m}x_{i,j}\widehat{p}_j = \sum_{j=1}^{m} \widehat{p}_j = \sum_{i=1}^{n} \sum_{j=1}^{m}y_{i,j}\widehat{p}_j.$$
	
	Consider an agent $i$ for which $\sum_{j=1}^{m}y_{i,j}\widehat{p}_j < \widehat{b}_i$.
	We argue that $x_i$ is not a best response to prices $\{\widehat{p}_j\}$.
	Indeed, agent $i$ can gain a higher utility by choosing $y_i$ and utilizing the remaining budget to improve her utility. Agent $i$ can do this because $y_i$ is not the entire set of items $I$, otherwise $y$ would be cycle-free.
	
	It follows that for every agent $i$, $\sum_{j=1}^{m}y_{i,j}\widehat{p}_j \geq \widehat{b}_i$.
	As $\sum_{i=1}^{n}\widehat{b}_i = \sum_{i=1}^{n} \sum_{j=1}^{m}y_{i,j}\widehat{p}_j$, we have that $\sum_{j=1}^{m}y_{i,j}\widehat{p}_j = \widehat{b}_i$.
	
	
	
\end{proof}


In conclusion:

\begin{proposition}
	For any Pareto optimal allocation there exist item prices and agent budgets (in tokens) supporting this allocation. Such prices and budgets can be computed efficiently.
\end{proposition}

\section{Max-Min Allocations}
\label{sec:maxmin}


Our pricing and budgeting mechanism assumes a desired Pareto optimal
allocation is already given\footnote{As before we assume that for every agent $i$ and item $l$ the valuation of agent $i$ for item $l$ is strictly greater than zero.}. 

One such allocation of interest is the max-min allocation, for which we can
compute the allocation efficiently. For the special case of two agents (Section \ref{two agents}) or two items (Section \ref{subsub:twoitems}) the algorithms become particularly simple and efficient. 

\subsection{Definition and efficient computation}

For a given an allocation $x$, the minimal agent valuation 
$W_{\min}(x)$ is
\[
W_{\min}(x)=\min_{i} \left\{\sum_j x_{i,j}v_{i,j}\right\}
\]

A max-min
allocation maximizes the minimal agent valuation.
For our setting of divisible items and additive valuations, it can
be computed efficiently. Specifically,
a max-min allocation is one that solves the following LP
maximization problem:
\begin{align}
    \max \lambda&\label{eq:lpmaxmin}\\
    \mbox{such that}\qquad&\nonumber\\
    \forall i\in A \qquad &\lambda \leq \sum_{j \in I} x_{i,j} v_{i,j} & \mbox{valuation for agent } i\nonumber\\
    \forall j\in I \qquad& \sum_{i\in A} x_{i,j}\leq 1 &\mbox{capacity per item}\nonumber\\
    &x_{i,j}\geq 0 \nonumber
\end{align}

\subsection{Properties of max-min allocations}


We first prove that in a max-min allocations, all agents are
guaranteed to have equal valuations.

\begin{lemma}
    \label{lem:equal-values}
    For any max-min allocation $x$, $\forall i, i' \in A, v_i(x_i) = v_{i'}(x_{i'})$.
\end{lemma}

\begin{proof}
    Proof by contradiction, assume that an allocation $x$ exists such that $x$ is a max-min allocation but not all valuations $v_i(x_i)$ are equal. Therefore, there exists some agent $k$ whose valuation $v_k(x_k)$ is not the minimal valuation. It follows that agent $k$ must have a non-zero valuation, and thus, there exists at least one item $l$ such that $x_{k,l}>0$.

    We now construct a new allocation $\hat{x}$ with minimal agent valuation greater than that of the minimal agent valuation in allocation $x$, thus contradicting the assumption that $x$ is a max-min allocation.
    For every agent $i$ and item $j \neq l$, we set $\hat{x}_{i,j} = x_{i,j}$. We set $\hat{x}_{k,l} = x_{k,l} - \epsilon$, where $\epsilon$ is sufficiently small such that $v_k(\hat{x}_k)$ is still strictly larger than the minimum agent valuation in allocation $x$.

    For every agent $i \neq k$ we set $\hat{x}_{i,l} = x_{i,l} + \frac{\epsilon}{n-1}$. Since all valuations $v_{i,l}$ are non-zero, for every agent $i \neq k$, $v_i(\hat{x}_i) > v_i(x_i)$. Therefore, the minimal agent valuation in allocation $\hat{x}$ is strictly larger than that of allocation $x$, in contradiction to the assumption that $x$ is a max-min allocation. 
\end{proof}

We now relate max-min allocations to Pareto
optimal allocations as follows.

\begin{lemma}
    \label{lem:max-min is pareto}
An allocation $x$ is a max-min allocation if and only if it is  Pareto optimal and all agents have the same valuation.
\end{lemma}

\begin{proof}
Let $x$ be a max-min allocation. By Lemma \ref{lem:equal-values} all agent valuations are equal, so we need to show that $x$ is Pareto optimal.
For contradiction, assume $x$ is a max-min allocation but not Pareto optimal. Therefore, there exists some allocation $z$ that Pareto dominates $x$.

By Lemma \ref{lem:equal-values}, all agents in $x$ have equal valuations. Since $z$ Pareto dominates $x$, the minimal agent valuation in allocation $z$ must be at least the same as the minimal agent valuation in allocation $x$, and therefore $z$ is also a max-min allocation.

Since $z$ Pareto dominates $x$, at least one agent has a higher valuation than the same agent in $x$. However, by Lemma \ref{lem:equal-values}, all agents in $z$ also have equal valuations, thus all agents in $z$ have greater valuations, in contradiction to the assumption that $x$ is a max-min allocation.

To show the other derivation, assume that $x$ is both Pareto optimal and the agents have an equal valuations. For contradiction assume that $x$ is not a max-min allocation. Then there is another allocation $x'$ which maximizes the min agent valuation. By Lemma \ref{lem:equal-values} all the agents have the same valuation in $x'$. This implies that all agents improve by moving from $x$ to $x'$, contradicting the Pareto optimality of $x$.
\end{proof}

Since we can compute a max-min allocation efficiently via linear
programming, by \Cref{thm:origalloc} we have:

\begin{theorem}
A max-min allocation can be efficiently computed along with item
prices and agent budgets such that the max-min allocation is in
the best response of each agent.
\end{theorem}

\subsection{Simple, fast max-min allocation algorithms for special cases}

In this section we show that for the case of two agents or two
items, we can have simpler and faster algorithms.


\subsubsection{Two agents, Multiple items} \label{two agents}

The idea is to sort the items based on the relative preference of
the two agents, and allocate each agent its more favorable items.
We construct an allocation where both agents have
the same valuation. Specifically, this is done as follows.
For  each item $j$ define $\phi_j=v_{2,j}/v_{1,j}$, which is the
relative preference for item $j$. Sort the items by $\phi_j$. Let
$\pi$ be a permutation of $1,\ldots,m$ such that $\phi_{\pi_j} \leq
\phi_{\pi_{j+1}}$ for $j=1,\ldots, m-1$. For simplicity we assume
that $\pi_j = j$.

Agent $1$ is assigned items in the order $\phi_{1},
\phi_{2}, \ldots$ whereas agent $2$ is assigned items in the reverse order, i.e.,
$\phi_{m}, \phi_{m-1}, \ldots$. We simultaneously increase the
allocations for both agents, so that the valuations to the two
agents are equal, until the two thresholds meet, producing a
partition of the items. The infinitesimal rate at which the value
for agent $1$ increases is set to be equal to the infinitesimal rate
at which the value for agent $2$ increases.
%
Namely, we find a ``median" item: an index $s$ such that
$\sum_{j=1}^s v_{1,j} \geq \sum_{j=s+1}^m v_{2,j}$ and
$\sum_{j=1}^{s-1} v_{1,i} \leq \sum_{j=s}^m v_{2,j}$.

Except for item $s$  all other items are assigned
in their entirety to one of the two agents.  Item $s$ may be partitioned amongst the two agents. Let this 
allocation be denoted $x_{i,j}$, where $x_i =
x_{i,1}, x_{i,2}, \ldots, x_{i,m}$ for
$i=1,2$. Note that
$$
x_{1}= 1,1,\ldots,1, x_{1,s}, 0,\ldots,0, \mbox{\rm\
whereas\ }x_{2}= 0,0,\ldots,0, 1-x_{1,s}, 1,\ldots,1.
$$
where $x_{1,s}$ satisfies the identity $v_{1,s}x_{1,s}+\sum_{j=1}^{s-1}v_{1,j}=v_{2,s}(1-x_{1,s})+\sum_{j=s+1}^{m}v_{2,j}$.
(It is easy to verify that indeed $x_{1,s}\in[0,1]$.)

As every item goes to the agent with higher value for the item such an allocation is Pareto optimal. Additionally, the allocation gives identical valuation to both agents. By \cref{lem:max-min is pareto}, it follows that the allocation is a max-min allocation. Moreover, finding the ``median" item can be done via sorting in $O(m\log m)$ time, this bound dominates the complexity of finding the allocation.

We now set the prices and budgets. Item $j$ is priced at
$p_j=v_{1,j}$. The budget for agent $1$ is set to $b_1=\sum_j
x_{1,j} v_{1,j}$.

We now argue that  $x_{1}$ is in the demand set of agent
$1$. The value per token for agent $1$ is equal across all items,
therefore the agent is indifferent as to which items are chosen, as
long as the prices add up to $b_1$ and may as well be happy with the
items with values in increasing order of $v_{2,j}/v_{1,j}$.

The budget for agent $2$, is set to be $b_2=\sum_j p_j x_{2,j} =
\sum_j v_{1,j} x_{2,j}$. Again, we need to show that $x_2$
is in the demand set of agent $2$. Agent $2$ greedily prefers to use
her budget on items with the largest $v_{2,j}/p_j$.  She is assigned
the items in order of decreasing $v_{2,j}/p_j=v_{2,j}/v_{1,j}$.
Ergo, $x_{2}$ is in the demand set of agent $2$.

We have established the following theorem, for the case of two
agents.
\begin{theorem}
For two agents, a max-min allocation can be computed in $O(m\log m)$
time, along with item prices and agents' budgets such that the
max-min allocation is in the best response of each agent.
\end{theorem}

\noindent{\bf Remark:} The pricing and budgets above are not unique for max-min allocations. There are alternative prices and budgets that
would work equally well. For example, the pricing
$p_j=\max(v_{1,j},v_{2,j})$ also works as well as  other pricing
schemes.

\subsubsection{ Multiple agents, two items}
\label{subsub:twoitems}
We now consider the case of two items with multiple agents. The idea
is rather similar to the two agents case. Each agent has a relative
preference between the two items. We will allocate fractions of each
item according to agent relative preferences. Each agent will
get a fraction of one of the two items, except for potentially one
agent which will receive a fraction from both items.

We define $\rho_i=v_{i,1}/v_{i,2}$, which is the relative preference
of agent $i$ between the two items. We sort the agents in decreasing
order of $\rho_i=v_{i,1}/v_{i,2}$. Similarly to subsection \ref{two
agents}, without loss of generality we assume the indices are
already ordered, i.e. $\rho_i\leq \rho_{i+1}$, $i=1,\ldots,n$.

The goal is to assign fractions of item $1$ to agents with larger
$\rho_i$ and fractions of item $2$ to agents with smaller $\rho_i$.
Agents will get equal value from their shares. There may be some
agent, agent $k$, that gets non zero fractions of both items.

To compute the allocation we define the following constraints.
For agent $i$, define variables $x_{i,1}, x_{i,2}$, $i=1,\ldots,n$, which are the
fractions that agent $i$ receives from item $1$ and $2$,
respectively. We seek a value of $k$ and a solution to the following
equations:
\begin{eqnarray}
                    x_{i,j}&\geq& 0; \qquad i=1,\ldots,n; \quad j=1,2 \label{eq:init1} \\
                    x_{1,2}=\cdots = x_{k-1,2} = x_{k+1,1}=\cdots =x_{n,1}&=& 0 \label{eq:init2} \\
                  x_{1,1}+x_{2,1}+\cdots+x_{k-1,1} + x_{k,1} &=& 1 \label{eq:equal1} \\
                                   x_{k,2}+x_{k+1,2}+\cdots+x_{n,2} &=& 1 \label{eq:equal2}\\
                                  x_{1,1}v_{1,1} = x_{2,1}v_{2,1} = \cdots =x_{k-1,1}v_{k-1,1} &=& x_{k,1}v_{k,1} + x_{k,2}v_{k,2}\label{eq:subs}\\ &=& x_{k+1,2}v_{k+1,2} = x_{k+2,2}v_{k+2,2} = \cdots =x_{n,2}v_{n,2}\nonumber
\end{eqnarray}
An allocation that satisfies Equations (\ref{eq:init1}) -- (\ref{eq:equal2}) is Pareto optimal since each agent receives a
fraction of her ``preferred'' item (note that the agent may really prefer the other item, but relatively less so than the other agents), and both items are fully allocated. Equation (\ref{eq:subs}) constrains all agents to have equal valuations. Thus, a solution to Equations (\ref{eq:init1}) -- (\ref{eq:subs}) is a Pareto optimal allocation in which all agent valuations are equal.

We now argue that such a solution exists.  By Lemma \ref{lem:max-min is pareto} such a solution is a max-min allocation, and we know that a max-min allocation exists as a solution to the linear program (\ref{eq:lpmaxmin}). It follows that there must be some value of $1\leq k\leq n$ so that Equations (\ref{eq:init1}) -- (\ref{eq:subs}) are satisfied.

\noindent {\bf Computational Efficency.} 
At first sight (\ref{eq:init1}, \ref{eq:init2},
\ref{eq:equal1}, \ref{eq:equal2})  might seem a complicated linear program. However, note
that once we fixed $k$, we can express all the variables using
$x_{k,1}$ and $x_{k,2}$. Namely, $x_{i,1}=(x_{k,1}v_{k,1} +
x_{k,2}v_{k,2})/v_{i,1}$, for $i\leq k-1$, and
$x_{i,2}=(x_{k,1}v_{k,1} + x_{k,2}v_{k,2})/v_{i,2}$, for $i\geq
k+1$. So essentially the identities (\ref{eq:equal1}) and
(\ref{eq:equal2}), have only two unknown, $x_{k,1}$ and $x_{k,2}$.
We solve the two equations with two unknown and verify that they are
non-negative, i.e., inequality (\ref{eq:init1}) holds. If it does,
we are done. This implies that for a given $k$, and a given order,
we can solve  (\ref{eq:init1}) -- (\ref{eq:equal2}) in linear time.

Instead of solving the linear program directly for
each possible value of $k$, one can directly solve equations
(\ref{eq:init2}) -- (\ref{eq:subs}),
doing a binary search for a value of $k$ for which a feasible solution exists
that also satisfies the constraints in (\ref{eq:init1}). When computing a given value of $k$, it cannot be the case that both $x_{k,1}$ and $x_{k,2}$ are negative. If $x_{k,1} < 0$, then $\sum_{i=1}^{k-1}x_{i,1}>1$, meaning the values of item $1$ for the first $k-1$ agents are insufficient to satisfy (\ref{eq:subs}) by any feasible partition, and some of them should switch to item $2$. The same holds for agents $[k+1,\ldots,n]$ and item $2$ if $x_{k,2} < 0$. Therefore, we need to search only over the part which has a negative value, which limits the search to $O(\log n)$ iterations. 

%

Given such a solution, we set $p_1=v_{k,1}$ and $p_2=v_{k,2}$. This
means that agent $k$ is indifferent between equal prized parts of
items $1$ and $2$, whereas all agents $i<k$ prefer equal prized parts of
item $1$ over the same prized parts of item $2$, and vice-versa for
agents $i>k$.

Agents $i \in \{1,\ldots,k-1\}$ receive budget $b_i = x_{i, 1} p_1$.
Agents $i \in \{k+1,\ldots,n\}$ receive budget $b_i = x_{i, 2} p_2$. The
$k$-th agent's budget is $b_k = x_{k, 1} p_1 + x_{k, 2}
 p_2$.
All agents $i<k$ can afford an $x_{i,1}$ fraction of item 1, the
item does not run out, and this is in their demand set. Likewise for
agents $i>k$ (with respect to item 2). Agent $k$ gets the leftovers
from both items, a $x_{k,1}$ fraction of item $1$ and a $x_{k,2}$ fraction
of item $2$, which is in her demand set.

We have established the following theorem.

\begin{theorem}
For two agents a max-min allocation can be computed in $O(n\log n)$
time,
along with item prices and agents' budgets in tokens such that the allocation
is in the best response of each agent.
\end{theorem}


\section{Discussion and Open Problems}
\label{sec:discussion}


In this paper we consider a market setting of agents with additive valuations over heterogeneous divisible items. 
We prove that given an arbitrary Pareto optimal allocation $x$, one can efficiently compute anonymous item prices and (possibly unequal) token budgets that support the allocation $x$; i.e., where every agent is maximally happy with her allocation $x_i$ given item prices and her budget. 
For the special case of max-min social welfare, we provide algorithms to compute the allocation itself in several cases of interest.

Extending our results to indivisible items: Obviously, any feasible allocation in the indivisible model is also feasible in the divisible model. However, a Pareto optimal allocation in the indivisible model need not be Pareto optimal in the divisible model. E.g., this may occur when agents gain from swapping carefully sized fractions of items but not items in whole.


Another obvious relaxation is not assuming that valuations are additive. For example, submodular, subadditive, unit demand, etc. Yet another issue is removing our assumption that all valuations are strictly greater than zero. 

As described above, the max-min allocation itself, as well as the supporting prices and budgets, can be computed using more efficient algorithms than in the general case. Even more so, if either the number of agents or the number of items is $2$. It remains an open question whether such results are possible for any constant number of agents or items.

Finally, in contrast to competitive equilibria with equal budgets, our model considers unequal budgets. 
Indeed, allowing unequal budgets is essential for supporting any Pareto optimal input allocation. 
Yet, it would be interesting to suggest some parameterization to the diversity of budgets, and study various problems as a function of such parameters. 
E.g., one can set a bound on the level of diversity and characterize the allocations that can be supported under the corresponding constraints.  

\newpage

\bibliographystyle{alpha}
\bibliography{bibliography}

\newcommand{\etalchar}[1]{$^{#1}$}
\begin{thebibliography}{CDG{\etalchar{+}}17}

\bibitem[AD54]{ArrowDebreu}
Kenneth~J. Arrow and Gerard Debreu.
\newblock Existence of an equilibrium for a competitive economy.
\newblock {\em Econometrica}, 22(3):265--290, 1954.

\bibitem[Arr51]{Arrow51}
Kenneth Arrow.
\newblock An extension of the basic theorems of classical welfare economics.
\newblock In {\em Second Berkeley Symposium on Mathematical Statistics and
  Probability}, 01 1951.

\bibitem[Bro11]{brouwer1911abbildung}
Luitzen Egbertus~Jan Brouwer.
\newblock {\"U}ber abbildung von mannigfaltigkeiten.
\newblock {\em Mathematische annalen}, 71(1):97--115, 1911.

\bibitem[Bro52]{brouwer1952intuitionist}
Luitzen Egbertus~Jan Brouwer.
\newblock An intuitionist correction of the fixed-point theorem on the sphere.
\newblock {\em Proceedings of the Royal Society of London. Series A.
  Mathematical and Physical Sciences}, 213(1112):1--2, 1952.

\bibitem[BS00]{BrainardScarf}
William~C. Brainard and Herbert~E. Scarf.
\newblock {How to Compute Equilibrium Prices in 1891}.
\newblock (1272), August 2000.

\bibitem[Bud11]{Budish2011}
Eric Budish.
\newblock {The Combinatorial Assignment Problem: Approximate Competitive
  Equilibrium from Equal Incomes}.
\newblock {\em Journal of Political Economy}, 119(6):1061--1103, 2011.

\bibitem[CDG{\etalchar{+}}17]{Cole17}
Richard Cole, Nikhil~R. Devanur, Vasilis Gkatzelis, Kamal Jain, Tung Mai,
  Vijay~V. Vazirani, and Sadra Yazdanbod.
\newblock Convex program duality, fisher markets, and nash social welfare.
\newblock In Constantinos Daskalakis, Moshe Babaioff, and Herv{\'{e}} Moulin,
  editors, {\em Proceedings of the 2017 {ACM} Conference on Economics and
  Computation, {EC} '17, Cambridge, MA, USA, June 26-30, 2017}, pages 459--460.
  {ACM}, 2017.

\bibitem[CV07]{CVfisher}
Bruno Codenotti and Kasturi Varadarajan.
\newblock Computation of market equilibria by convex programming.
\newblock In Noam Nisan, Tim Roughgarden, {\'{E}}va Tardos, and Vijay~V.
  Vazirani, editors, {\em Algorithmic Game Theory}, chapter~6, pages 135--158.
  Cambridge University Press, 2007.

\bibitem[Deb51]{Debreu51}
Gerard Debreu.
\newblock The coefficient of resource utilization.
\newblock {\em Econometrica}, 19(3):273--292, 1951.

\bibitem[GS99]{GulStacccetti}
Faruk Gul and Ennio Stacchetti.
\newblock {Walrasian Equilibrium with Gross Substitutes}.
\newblock {\em Journal of Economic Theory}, 87(1):95--124, July 1999.

\bibitem[KJC82]{kelso1982job}
Alexander~S Kelso~Jr and Vincent~P Crawford.
\newblock Job matching, coalition formation, and gross substitutes.
\newblock {\em Econometrica: Journal of the Econometric Society}, pages
  1483--1504, 1982.

\bibitem[Lem17]{renatosurvey}
Renato~Paes Leme.
\newblock Gross substitutability: An algorithmic survey.
\newblock {\em Games Econ. Behav.}, 106:294--316, 2017.

\bibitem[Raw99]{Rawls99}
John Rawls.
\newblock {\em A Theory of Justice}.
\newblock Harvard University Press, 1999.

\bibitem[Sti81]{Stiglitz81}
Joseph~E. Stiglitz.
\newblock Pareto optimality and competition.
\newblock {\em The Journal of Finance}, 36(2):235--251, 1981.

\bibitem[Var74]{Varian74}
Hal Varian.
\newblock Equity, envy, and efficiency.
\newblock {\em Journal of Economic Theory}, 9(1):63--91, 1974.

\bibitem[Vaz07]{Vajiranifisher}
Vijay~V. Vazirani.
\newblock Combinatorial algorithms for market equilibria.
\newblock In Noam Nisan, Tim Roughgarden, {\'{E}}va Tardos, and Vijay~V.
  Vazirani, editors, {\em Algorithmic Game Theory}, chapter~5, pages 103--134.
  Cambridge University Press, 2007.

\bibitem[Wal74]{walras1874}
L.~Walras.
\newblock {\em {\'E}l{\'e}ments d'{\'e}conomie politique pure; ou, Th{\'e}orie
  de la richesse sociale}.
\newblock Corbaz, 1874.

\end{thebibliography}

\end{document}